\documentclass[a4paper, reqno]{amsart}
\usepackage{psfig}
\usepackage{graphicx}
\usepackage{array}
\newcolumntype{M}[1]{>{\raggedright}m{#1}}
\usepackage{tikz}
\usetikzlibrary{calc}

\newtheorem{theorem}{Theorem}
\newtheorem{lemma}[theorem]{Lemma}

\newtheorem{proposition}[theorem]{Proposition}
\newtheorem{corollary}[theorem]{Corollary}

\newtheorem{Fact}[theorem]{Fact}
\theoremstyle{definition}
\newtheorem{definition}{Definition}

\newcommand{\btree}{\ensuremath{\beta(0,1)}-tree}
\newcommand{\btrees}{\btree s}
\newcommand{\F}{\mathcal{F}}

\newcommand{\nodestyle}{
  \tikzstyle{every node} = [font=\footnotesize];
}
\newcommand{\discstyle}{
  \tikzstyle{blk} =
    [ circle,fill=black,draw=black, minimum size=3.7pt, inner sep=0pt ];
  \tikzstyle{wht} =
    [ circle,fill=white,draw=black, minimum size=3.7pt, inner sep=0pt ];
  \tikzstyle{sml} =
    [ circle,fill=black,draw=black, minimum size=3pt, inner sep=0pt ];
}

\newcommand{\style}{
  \nodestyle
  \discstyle
}
\newcommand{\scl}{0.5}

\newcommand{\trea}[4]{
  \begin{tikzpicture}[scale=\scl, semithick, baseline=(d)]
    \style;
    \node [sml] (a) at (0,0) {};
    \node [sml] (b) at (0,1) {};
    \node [sml] (c) at (0,2) {};
    \node [sml] (d) at (0,3) {};
    \draw 
      (a) node[right=2pt] {#1} --
      (b) node[right=2pt] {#2} --
      (c) node[right=2pt] {#3} --
      (d) node[right=2pt] {#4};
  \end{tikzpicture}
}
\newcommand{\treav}[3]{
  \begin{tikzpicture}[scale=\scl, semithick, baseline=(d)]
    \style;
    \node [sml] (a) at (0,0) {};
    \node [sml] (b) at (0,1) {};
    \node [sml] (c) at (0,2) {};
    \draw 
      (a) node[right=2pt] {#1} --
      (b) node[right=2pt] {#2} --
      (d) node[right=2pt] {#3};
  \end{tikzpicture}
}
\newcommand{\treb}[4]{
  \begin{tikzpicture}[scale=\scl, semithick, baseline=(d)]
    \style;
    \node [sml] (a) at (-0.65,0) {};
    \node [sml] (b) at ( 0.65,0) {};
    \node [sml] (c) at (0,1) {};
    \node [sml] (d) at (0,2) {};
    \draw (a) node[below=2pt] {#1} -- (c);
    \draw (b) node[below=2pt] {#2} -- (c);
    \draw (c) node[right=2pt] {#3} -- (d) node[right=2pt] {#4};
  \end{tikzpicture}
}
\newcommand{\trebv}[3]{
  \begin{tikzpicture}[scale=\scl, semithick, baseline=(d)]
    \style;
    \node [sml] (a) at (-0.65,0) {};
    \node [sml] (b) at ( 0.65,0) {};
    \node [sml] (c) at (0,1) {};
    \draw (a) node[below=2pt] {#1} -- (c);
    \draw (b) node[below=2pt] {#2} -- (c) node[right=2pt] {#3};
  \end{tikzpicture}
}
\newcommand{\trec}[4]{
  \begin{tikzpicture}[scale=\scl, semithick, baseline=(d)]
    \style;
    \node [sml] (a) at (-0.65,0) {};
    \node [sml] (b) at (-0.65,1) {};
    \node [sml] (c) at ( 0.65,1) {};
    \node [sml] (d) at (   0,2) {};
    \draw (a) node[left=2pt] {#1} -- (b);
    \draw (b) node[left=2pt] {#2} -- (d);
    \draw (c) node[below=1pt] {#3} -- (d) node[right=2pt] {#4};
  \end{tikzpicture}
}
\newcommand{\tred}[4]{
  \begin{tikzpicture}[scale=\scl, semithick, baseline=(d)]
    \style;
    \node [sml] (a) at ( 0.65,0) {};
    \node [sml] (b) at ( 0.65,1) {};
    \node [sml] (c) at (-0.65,1) {};
    \node [sml] (d) at (   0,2) {};
    \draw (a) node[right=2pt] {#1} -- (b);
    \draw (b) node[right=2pt] {#2} -- (d);
    \draw (c) node[below=1pt] {#3} -- (d) node[right=2pt] {#4};
  \end{tikzpicture}
}
\newcommand{\tree}[4]{
  \begin{tikzpicture}[scale=\scl, semithick, baseline=(d)]
    \style;
    \node [sml] (a) at (-1,0) {};
    \node [sml] (b) at ( 0,0) {};
    \node [sml] (c) at ( 1,0) {};
    \node [sml] (d) at ( 0,1) {};
    \draw (a) node[below=1pt] {#1} -- (d);
    \draw (b) node[below=1pt] {#2} -- (d);
    \draw (c) node[below=1pt] {#3} -- (d) node[right=2pt, yshift=2pt] {#4};
  \end{tikzpicture}
}
\newcommand{\figtree}{
  \begin{tikzpicture}[xscale=0.37, yscale=0.5, semithick, baseline=(r)]
    \style;
    \node [sml] (r)       at ( 0,6) {};
    \node [sml] (r1)      at (-3,5) {};
    \node [sml] (r2)      at (-1,5) {};
    \node [sml] (r3)      at ( 1,5) {};
    \node [sml] (r31)     at ( 1,4) {};
    \node [sml] (r4)      at ( 3,5) {};
    \node [sml] (r41)     at ( 3,4) {};
    \node [sml] (r411)    at ( 3,3) {};
    \node [sml] (r4111)   at ( 2,2) {};
    \node [sml] (r41111)  at ( 2,1) {};
    \node [sml] (r411111) at ( 2,0) {};
    \node [sml] (r4112)   at ( 4,2) {};
    \draw (r) node[above=2pt] {4} -- (r1)  node[below left=-1pt] {0}
          (r)                     -- (r2)  node[below left=-1pt] {0}
          (r)                     -- (r3)  node[left=2pt, yshift=-2pt] {1}
          (r)                     -- (r4)  node[right=2pt] {2}
          (r3)                    -- (r31) node[left=2pt] {0}
          (r4)                    -- (r41) node[right=2pt] {1}
          (r41)                   -- (r411) node[right=2pt] {3}
          (r411)                  -- (r4111) node[left=2pt] {2}
          (r4111)                 -- (r41111) node[left=2pt] {1}
          (r41111)                -- (r411111) node[left=2pt] {0}
          (r411)                  -- (r4112) node[right=2pt] {0};
  \end{tikzpicture}
}

\def\triangle[#1,#2]#3;{
  \begin{scope}[shift={#3}]
    \begin{scope}[rotate=#2]
      \path (0,0)     coordinate (O);
      \path (245:4.4) coordinate (A);
      \path (-65:4.4) coordinate (B);
      \filldraw[gray!20] (O) -- (A) -- (B) -- cycle;
      \draw (O) -- (A);
      \draw (O) -- (B);
      \node at (0,-2.6) {#1};
      \node[o] at (O) {};
    \end{scope}
  \end{scope}
}

\renewcommand{\root}{\operatorname{root}}

\DeclareMathOperator{\sub}{\mathrm{sub}}

\newcommand{\rootM}{\operatorname{f\mathfrak{1}r\mathfrak{3}}}
\newcommand{\subM}{\operatorname{s\mathfrak{1}r\mathfrak{3}}}

\title{Gray coding planar maps}

\author[S. Avgustinovich]{Sergey Avgustinovich}
\address{
  Sobolev Institute of Mathematics, 4 Acad. Koptyug Ave,
  630090 Novosibirsk, Russia,
  \texttt{avgust@math.nsc.ru}
}
\author[S. Kitaev]{Sergey Kitaev}
\address{
  Department of Computer and Information Sciences,
  University of Strathclyde,
  26 Richmond Street,
  Glasgow G1 1XH, United Kingdom,
  \texttt{sergey.kitaev@strath.ac.uk}
}
\author[V. N. Potapov]{Vladimir N. Potapov}
\address{
  Sobolev Institute of Mathematics, 4 Acad. Koptyug Ave,
  630090 Novosibirsk, Russia,
  \texttt{vpotapov@math.nsc.ru}
}
\author[V. Vajnovszki]{Vincent Vajnovszki}
\address{ 
LE2I, Universit\'e de Bourgogne Franche-Comt\'e, 
BP 47870, 21078 Dijon Cedex, France, 
  \texttt{vvajnov@u-bourgogne.fr}
}

\begin{document}

\begin{abstract}
The idea of (combinatorial) Gray codes is to list objects in question in such a way that two successive objects
differ in some pre-specified small way.
In this paper, we utilize  $\beta$-description trees to cyclicly Gray code three classes of cubic planar maps, namely, bicubic planar maps, 3-connected cubic planar maps, and cubic non-separable planar maps.
\end{abstract}


\maketitle
\thispagestyle{empty}

\section{Introduction}

\noindent
{\bf Gray codes.} The problem of exhaustively listing the objects of a given class is important for several fields of science such as computer science, hardware and software, biology and (bio)chemistry. The idea of so-called {\em Gray codes} (or {\em combinatorial Gray codes}) is to list the objects in such a way that two successive objects differ in some pre-specified small way; in addition, if the last and first objects in the list differ in the same small way, then the Gray code is called cyclic. 
In \cite{Walsh}  a general definition is given, where a Gray code is defined as ‘an infinite set of word-lists with unbounded word-length such that the Hamming distance between any two successive words is bounded independently of the word-length’ (the Hamming distance is the number of positions in which the words differ).

Originally, a Gray code was used in a telegraph demonstrated by French engineer \'Emile Baudot in 1878. However, these days we normally say ``the Gray code'' to refer to the {\em reflected binary code} introduced by Frank Gray in 1947 to list all binary words of length $n$. Much has been discovered and written about the Gray code (see for example
\cite{Rus} or
\cite{Savage,Doran} for surveys) and it was used, for example, in {\em error corrections in digital communication} and in solving puzzles as {\em Tower of Hanoi puzzle}. 
On the other hand, the area of combinatorial Gray codes was popularized by Herbert Wilf in 1988-89 and since then such codes were found for many combinatorial structures, e.g. for {\em involutions} and {\em fixed-point free involutions}, {\em derangements} and certain classes of {\em pattern-avoiding permutations} (see \cite{DFMV,Savage} and references therein).

Existence of a (resp., cyclic) Gray code is often established via finding a {\em Hamiltonian path} (resp., {\em Hamiltonian cycle}) in a certain graph corresponding to objects in question. In such a graph two vertices are connected by an edge if and only if the respective objects can follow each other in a Gray code. A Hamiltonian path (resp., Hamiltonian cycle) in a graph is a path (resp., cycle) in the graph that goes through each vertex exactly once. \\    

\noindent
{\bf Planar maps.} A \emph{planar map} is a connected graph embedded in the plane with
no edge-crossings, considered up to continuous deformations. A map has
\emph{vertices}, \emph{edges}, and \emph{faces}. The maps we consider shall be \emph{rooted},
meaning that a directed edge has been distinguished as the root. Without loss of generality, we assume that the root is always on the outer-face, called {\em root face}, and it is oriented counterclockwise.

A planar map in which each vertex is of degree 3 is {\em cubic}; it is
{\em bicubic} if, in addition, it is bipartite, that is, if its
vertices can be colored using two colors, say, black and white, so
that adjacent vertices are assigned different colors. A map is {\em $k$-connected} if there does not exist a set of $k-1$ vertices whose removal disconnects the map. $2$-connected maps are also known as {\em non-separable maps}.

For brevity, we omit the word ``planar'' in the classes of planar maps considered in this paper. 

Tutte \cite[Chapter 10]{Tutte1998}
founded the enumerative theory of planar maps in a series of papers in
the 1960s (see~\cite{Tutte1963} and the references in~\cite{CS1997}). In particular, 
the number of bicubic maps and cubic non-separable maps on $2n$ vertices are, respectively, $$\frac{3\cdot 2^{n-1}(2n)!}{n!(n+2)!}\mbox{\ \ and \ \ }\frac{2^n(3n)!}{(n+1)!(2n+1)!}.$$ 

\ \\

\noindent
{\bf $\beta(a,b)$-trees and planar maps.} A {\em valuated tree} is a rooted plane tree with non-negative integer labels on its vertices. A {\em description tree} introduced by Cori et al. in~\cite{CJS} is a valuated tree such that the label of each vertex $v$ belongs to a set of values that depends only on the labels of $v$'s sons according to a given rule. Description trees give a framework for recursively decomposing several families of planar
maps. $\beta$-description trees, introduced next, are of interest in this paper.

\begin{definition}\label{betaAB-tree} A $\beta(a,b)$-tree is a rooted plane tree whose vertices are labeled with non-negative integers such that
\begin{enumerate}
\item leaves have label $a$;
\item the label of the root is the sum of its children's
  labels;
\item the label of any other vertex is at least $a$ and at most $b$ plus the sum of its children's
  labels.
\end{enumerate}
\end{definition}

It was shown in~\cite{CJS,CS1997} that the following objects are in one-to-one correspondence:
\begin{itemize}
\item $\beta(0,1)$-trees and bicubic maps;
\item $\beta(1,1)$-trees and 3-connected cubic maps;
\item $\beta(2,2)$-trees and cubic non-separable maps; and
\item $\beta(1,0)$-trees and non-separable maps.
\end{itemize}

Also, it is straightforward to see that $\beta(0,0)$-maps are in one-to-one correspondence with rooted planar trees, since one can erase the labels in this case as all of them are 0. \\

\noindent
{\bf The main results in this paper.} One can ask the following question: Is it possible to Gray code a given class of maps? To our best knowledge, no results are known in this direction possibly due to a rather complicated nature of (planar) maps.  Thus, one needs to encode the class of maps by words, and then to try to list these words using specified criteria on the number of positions in which the words can differ. Our approach is in encoding the $\beta$-description trees in question, which are in a bijective correspondence with the maps of interest, on $n$ vertices  by tuples of length $3n-2$; the first $2n-2$ elements of the tuple encode the shape of a tree (using so-called {\em Dyck words}), and the remaining elements are used to encode its labels. In either case, for convenience of presentation, we will consider Gray coding shapes of trees separately, which will be given by a known result, while a real challenge will be lying in (cyclicly) Gray coding $\beta$-description trees having the same shape.    

We note that $\beta$-description trees have already been used to obtain non-trivial equidistribution results on planar maps, e.g. bicubic maps~\cite{CKM}, and these trees are a key object in this paper. We will present our results on \btrees, which will give a Gray code for bicubic maps, and then discuss a straightforward extension of that to $\beta(a,b)$-trees with $b\geq 1$. The later will give at once Gray codes for cubic non-separable maps and 3-connected cubic maps. Thus, our focus will be on \btrees. In particular, the only bijective correspondence we will explain in this paper is that between \btrees\ and bicubic maps, to give an idea on how bijections between maps and $\beta$-description trees corresponding to them could look like; we refer to \cite{CJS,CS1997} for bijections between $\beta(1,1)$-trees  (resp., $\beta(2,2)$-trees) and 3-connected cubic maps (resp., cubic non-separable maps).

This paper is organized as follows. In Section~\ref{sec:btrees} we discuss \btrees\ and bicubic maps, in particular sketching a bijection between these sets of objects. In Section~\ref{btrees-same-underlying} we discuss a key component in this paper, namely, Gray coding \btrees\ having the same shape. Cyclic Gray coding \btrees\ having the same shape is discussed in Section~\ref{btrees-same-underlying-cyclic}. Even though Gray coding cyclicly is what we are actually interested in, we first present a Gray code for \btrees\ having the same shape without the cyclic requirement to prepare the Reader for the more involved arguments in the cyclic case. The main results are presented in Section~\ref{main} along with a definition of Dyck words and necessary results about them. Finally, in Section~\ref{final-sec} we provide several directions for further research.  

\section{\btrees\ and bicubic maps}\label{sec:btrees}

Letting $a=0$ and $b=1$ in Definition~\ref{betaAB-tree} we will obtain a definition of a \btree. Note that the label of the root of a \btree\ is defined uniquely from the labels of its children, which allows us to modify this definition to suit better to our purposes. In this paper, we will consider two modifications of the definition. First, we will re-define the root label to be one more than the sum of its children (as was done in \cite{CKM} for a better description of statistics preserved under the bijection with bicubic maps to be described below), and then we will let the root label be $*$ (to allow two \btrees\ having the same shape to differ just in one label). Thus, no matter which definition we use, we still have a class of trees in one-to-one correspondence with the originally defined \btrees, and slightly abusing the notation, which will not cause any confusion, we will refer to all of the ``modified \btrees'' as \btrees.    

We continue with stating a slightly modified definition of \btrees, which are particular
instances of $\beta(a,b)$-trees introduced in Definition \ref{betaAB-tree}.

\begin{definition}\label{beta01-tree} A \btree\ is a rooted plane tree whose vertices  are labeled with
nonnegative integers such that
\begin{enumerate}
\item leaves have label $0$;
\item the label of the root is one more than the sum of its children's
  labels;
\item the label of any other vertex exceeds the sum of its children's
  labels by at most 1.
\end{enumerate}
\end{definition}

The unique \btree\ with exactly one vertex (and no edges) is called
{\em trivial}; the root of the trivial tree has label  $0$. Any other \btree\ is called {\em nontrivial}. In
Figure~\ref{beta01}, appearing in~\cite{CKM}, we have listed all \btrees\ on 4 vertices.
\begin{figure}[h]
  $$
  \trea{0}{0}{0}{1}\qquad
  \trea{0}{0}{1}{2}\qquad
  \trea{0}{1}{0}{1}\qquad
  \trea{0}{1}{1}{2}\qquad
  \trea{0}{1}{2}{3}\qquad
  \treb{0}{0}{0}{1}\qquad
  \treb{0}{0}{1}{2}
  $$
  $$
  \trec{0}{0}{0}{1}\qquad
  \trec{0}{1}{0}{2}\qquad
  \tred{0}{0}{0}{1}\qquad
  \tred{0}{1}{0}{2}\qquad
  \tree{0}{0}{0}{1}
  $$
  \caption{All \btrees\ on 4 vertices.}\label{beta01}
\end{figure}
Let $\root(T)$ denote the root label of $T$, and let $\sub(T)$ denote
the number of children of the root.  We say that a \btree\ $T$ is {\em
  reducible} if $\sub(T)>1$, and {\em irreducible} otherwise. Any
reducible tree can be written as a sum of irreducible ones, where the
sum $U\oplus V$ of two trees $U$ and $V$ is defined as the tree
obtained by identifying the roots of $U$ and $V$ into a new root with
label $\root(U)+\root(V)-1$.  See Figure~\ref{fig:tree}, taken from~\cite{CKM}, for an
example.
\begin{figure}[h]
  \figtree
  \!\!\!\!=\;\;\;\;
  \begin{tikzpicture}[xscale=0.37, yscale=0.5, semithick, baseline=(r)]
    \style;
    \node [sml] (r)       at (0,1) {};
    \node [sml] (r1)      at (0,0) {};
    \draw (r) node[above=2pt] {1} -- (r1) node[below=2pt] {0};
  \end{tikzpicture}
  \,$\oplus$\,
  \begin{tikzpicture}[xscale=0.37, yscale=0.5, semithick, baseline=(r)]
    \style;
    \node [sml] (r)       at (0,1) {};
    \node [sml] (r1)      at (0,0) {};
    \draw (r) node[above=2pt] {1} -- (r1) node[below=2pt] {0};
  \end{tikzpicture}
  \,$\oplus$\,
  \begin{tikzpicture}[xscale=0.37, yscale=0.5, semithick, baseline=(r)]
    \style;
    \node [sml] (r)       at (0,2) {};
    \node [sml] (r3)      at (0,1) {};
    \node [sml] (r31)     at (0,0) {};
    \draw (r) node[above=2pt] {2} -- (r3) node[right=1pt] {1} -- (r31) node[below=2pt] {0};
  \end{tikzpicture}
  $\!\oplus\!$\hspace{-2.5ex}
  \begin{tikzpicture}[xscale=0.37, yscale=0.5, semithick, baseline=(r)]
    \style;
    \node [sml] (r)       at (3,6) {};
    \node [sml] (r4)      at (3,5) {};
    \node [sml] (r41)     at (3,4) {};
    \node [sml] (r411)    at (3,3) {};
    \node [sml] (r4111)   at (2,2) {};
    \node [sml] (r41111)  at (2,1) {};
    \node [sml] (r411111) at (2,0) {};
    \node [sml] (r4112)   at (4,2) {};
    \draw (r) node[above=2pt] {3}
          (r)                     -- (r4)  node[right=2pt] {2}
          (r4)                    -- (r41) node[right=2pt] {1}
          (r41)                   -- (r411) node[right=2pt] {3}
          (r411)                  -- (r4111) node[left=2pt] {2}
          (r4111)                 -- (r41111) node[left=2pt] {1}
          (r41111)                -- (r411111) node[left=2pt] {0}
          (r411)                  -- (r4112) node[right=2pt] {0};
  \end{tikzpicture}
  \caption{Decomposing a reducible \btree.}\label{fig:tree}
\end{figure}

Note also that any irreducible tree with at least one edge
is of the form $\lambda_i(T)$, where $0 \leq i\leq \root(T)$ and
$\lambda_i(T)$ is obtained from $T$ by joining a new root via an
edge to the old root; the old root is given the label $i$, and the new
root is given the label $i+1$. For instance,
$$
\text{if }\;
T =
\begin{tikzpicture}[xscale=0.37, yscale=0.5, semithick, baseline=(r)]
  \style;
  \node [sml] (r)  at ( 0,2) {};
  \node [sml] (r1) at (-1,1) {};
  \node [sml] (r2) at ( 1,1) {};
  \node [sml] (r21) at ( 1,0) {};
  \draw (r) node[above=2pt] {2} -- (r1) node[below=2pt] {0}
        (r)                     -- (r2) node[right=1pt] {1} -- (r21) node[below=2pt] {0};
\end{tikzpicture}
\text{then}\quad
\newcommand{\lambdaT}[2]{
\begin{tikzpicture}[xscale=0.37, yscale=0.5, semithick, baseline=(r')]
  \style;
  \node [sml] (r') at ( 0,3) {};
  \node [sml] (r)  at ( 0,2) {};
  \node [sml] (r1) at (-1,1) {};
  \node [sml] (r2) at ( 1,1) {};
  \node [sml] (r21) at ( 1,0) {};
  \draw (r') node[above=2pt] {#2} -- (r);
  \draw (r)  node[right=2pt] {#1} -- (r1) node[below=2pt] {0}
        (r)                       -- (r2) node[right=1pt] {1} -- (r21) node[below=2pt] {0};
\end{tikzpicture}
}
\lambda_0(T) = \!\!\!\!\lambdaT{0}{1}\hspace{-3.5ex},\quad
\lambda_1(T) = \!\!\!\!\lambdaT{1}{2}\hspace{-3.5ex},\quad \text{and }\;
\lambda_2(T) = \!\!\!\!\lambdaT{2}{3}\hspace{-3.5ex}.\quad
$$

The smallest bicubic map has two vertices and three edges joining them. It is well-known that the faces of a bicubic map
can be colored using three colors so that adjacent faces have distinct
colors, say, colors 1, 2 and 3, in a counterclockwise order around
white vertices. We will assume that the root vertex is black and the
root face has color 3. There are exactly three different bicubic maps
with 6 edges and they are given in Figure~\ref{bicubic} appearing in~\cite{CKM}.

\begin{figure}[h]
  \begin{tikzpicture}[semithick, scale=0.80, baseline=(c.base), bend angle=85, >=stealth]
    \style;
    \node [wht] (a) at (0,0) {};
    \node [blk] (b) at (1,0) {};
    \node [wht] (c) at (2,0) {};
    \node [blk] (d) at (3,0) {};
    \node (n1) at (1.5,-0.8) {1};
    \node (n2) at (1.5, 0  ) {3};
    \node (n3) at (1.5, 0.8) {2};
    \node (n4) at (2.85, 1.15) {3};
    \path (a) edge[bend right, looseness=1.3] (d);
    \path (d) edge[bend right, looseness=1.3, shorten >=0.15pt, ->] (a);
    \path (b) edge[bend right, looseness=1.2] (c);
    \path (c) edge[bend right, looseness=1.2] (b);
    \path (a) edge (b) (c) edge (d);
  \end{tikzpicture}
  \qquad\quad
  \begin{tikzpicture}[semithick, scale=0.95, baseline=(n2), bend angle=85, >=stealth]
    \style;
    \node [blk] (a) at (0,0) {};
    \node [wht] (b) at (1,0) {};
    \node [wht] (c) at (0,1) {};
    \node [blk] (d) at (1,1) {};
    \node (n1) at (0.5,-0.25) {2};
    \node (n2) at (0.5, 0.50) {1};
    \node (n3) at (0.5, 1.25) {2};
    \node (n4) at (1.3, 0.50) {3};
    \path (a) edge (b) (b) edge (d) (d) edge (c) (c) edge (a);
    \path (a) edge[bend right, looseness=1.6] (b);
    \path (d) edge[bend right, looseness=1.6, shorten >=0.15pt, ->] (c);
  \end{tikzpicture}
  \qquad\quad
  \begin{tikzpicture}[semithick, scale=0.95, baseline=(n2), bend angle=85, >=stealth]
    \style;
    \node [wht] (a) at (0,0) {};
    \node [blk] (b) at (1,0) {};
    \node [blk] (c) at (0,1) {};
    \node [wht] (d) at (1,1) {};
    \node (n1) at (0.5,-0.25) {1};
    \node (n2) at (0.5, 0.50) {2};
    \node (n3) at (0.5, 1.25) {1};
    \node (n4) at (1.3, 0.50) {3};
    \path (a) edge (b) (b) edge (d) (d) edge (c) (c) edge[shorten >=0.15pt, ->] (a);
    \path (a) edge[bend right, looseness=1.6] (b);
    \path (d) edge[bend right, looseness=1.6] (c);
  \end{tikzpicture}
  \caption{All bicubic maps with 4 vertices.}\label{bicubic}
\end{figure}
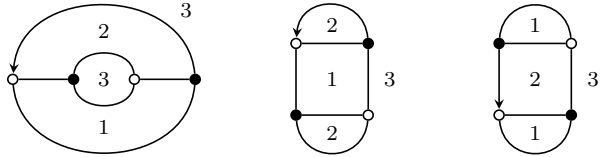

Following \cite{CKM} we will now describe a bijection between 
bicubic maps and \btrees. For any bicubic map $M$ and $i=1,2,3$, let
$\F_i(M)$ be the set of $i$-colored faces of $M$. Let
$R_1\in\F_1(M)$, $R_2\in\F_2(M)$, and $R_3\in\F_3(M)$ be the three
faces around the root vertex; in particular, $R_3$
is the root face. In addition, let $S_1\in\F_1(M)$ be the $1$-colored
face that meets the vertex that the root edge points at:
$$
\begin{tikzpicture}[semithick, scale=0.7, baseline=-0.6ex]
  \style;
  \node [blk] (r) at (0,0) {};
  \node [wht] (s) at (1,0) {};
  \path (-1,0) edge (r) edge[->] (s) edge (2,0);
  \path (r) edge (0,1);
  \path (s) edge (1,1);
  \node [wht] (s) at (1,0) {};
  \node at (-0.5, 0.5) {$R_1$};
  \node at ( 0.5, 0.5) {$R_2$};
  \node at ( 1.5, 0.5) {$S_1$};
  \node at ( 0.5,-0.5) {$R_3$};
\end{tikzpicture}
$$ 
Let us say that a face touches another face $k$ times if there are $k$
different edges each belonging to the boundaries of both faces.
Define the following two statistics:
\begin{align*}
\rootM(M) &\;\;\text{is the number of faces in $\F_1(M)$ that touch $R_3$;}\\  
\subM(M)  &\;\;\text{is the number of times $S_1$ touches $R_3$.}
\end{align*}
We say that $M$ is {\em irreducible} if $\subM(M)=1$, or, in other
words, if $S_1$ touches $R_3$ exactly once; we say that $M$ is {\em
  reducible} otherwise. We shall introduce operations on bicubic
maps that correspond to $\lambda_i$ and $\oplus$ of \btrees.  This
will induce the desired bijection $\psi$ between bicubic maps and
$\beta(0,1)$-trees.

To construct an irreducible bicubic map based on $M$, and having two
more vertices than $M$, we proceed in one of two ways. The first way (1)
corresponds to  $\lambda_i(T)$ when $i=\root(T)$; the second way (2)
corresponds to $\lambda_i(T)$ when $0\leq i < \root(T)$.

\def\mapM{ \filldraw[fill=blue!6,draw=black!50] (0,0) circle (1);
  \node[font=\small] at (0,0) {$M$}; }

\begin{itemize}
\item[(1)] We create a new $1$-colored face touching the root face
  exactly once, so $\rootM(M')=\rootM(M)+1$, by removing the root edge
  from $M$ and adding a digon that we connect to the map as in the
  picture below.
  $$
  \begin{tikzpicture}[semithick, scale=0.65, baseline=-0.6ex]
    \style;
    \mapM;
    \draw [->] (240:1) arc (240:291:1);
    \node [blk] (r) at (240:1) {};
    \node [wht] (s) at (300:1) {};
    \node at (-1.4,0) {\tiny{3}};
  \end{tikzpicture}
  \quad \longmapsto \quad M' \,=\; 
  \begin{tikzpicture}[semithick, scale=0.65, baseline=-0.6ex]
    \style;
    \mapM;
    \draw[white, thick] (240:1) arc (240:300:1);
    \draw (240:1) -- +(0,-1) coordinate (a);
    \draw[->] (240:1) -- +(0,-0.84);
    \draw (300:1) -- +(0,-1) coordinate (b);
    \path (a) edge[bend right, looseness=1.5] (b);
    \path (a) edge[bend left, looseness=1.5] (b);
    \node [wht] at (a) {};
    \node [blk] at (b) {};
    \node [blk] (r) at (240:1) {};
    \node [wht] (s) at (300:1) {};
    \node at ($ (a) !.5! (b) $) {\tiny{1}};
    \node[yshift=0.5pt] at ($ (r) !.5! (b) $) {\tiny{2}};
    \node at (1.4,0) {\tiny{3}};
  \end{tikzpicture}
  $$
\item[(2)] Assuming that $\rootM(M)=k$; that is, $M$ has $k$ (different) $1$-colored
  faces touching the root face, we can create an irreducible map $M'$
  such that $\rootM(M')=i$, where $1\leq i\leq k$.  To this end, we
  remove the root edge from $M$. Starting at the root vertex and
  counting in {\em clockwise direction}, we also remove the first edge
  of the $i$th $1$-colored face that touches the root face. In the
  picture below we schematically illustrate the case $i=3$. Next we add
  two more vertices and respective edges, and assign a new root as shown in the
  figure.
  \def\radius{1.5}
  \def\radiux{1.3}
  $$ 
  \begin{tikzpicture}[semithick, scale=0.65, bend angle=45, rotate=0,  baseline=-0.6ex] 
    \style;
    \path
    coordinate (P0) at (0*45:\radius) {}
    coordinate (P1) at (1*45:\radius) {}
    coordinate (P2) at (2*45:\radius) {}
    coordinate (P3) at (3*45:\radius) {}
    coordinate (P4) at (4*45:\radius) {}
    coordinate (P5) at (5*45:\radius) {}
    coordinate (P6) at (6*45:\radius) {}
    coordinate (P7) at (7*45:\radius) {};
    \filldraw[fill=blue!6,draw=black!50] (0,0) circle (\radius);
    \node[font=\small] at (0,0) {$M$};
    \path (P1) edge[bend left, looseness=1.5] (P2);
    \path (P3) edge[bend left, looseness=1.5] (P4);
    \path (P5) edge[bend left, looseness=1.5] (P6);
    \path (P7) edge[bend left, looseness=1.5] (P0);
    \draw [->] (P6) arc (6*45:7*45-5:\radius);
    \node at (-2,0) {\tiny{3}};
    \node at (1.5*45:\radiux) {\tiny{1}};
    \node at (3.5*45:\radiux) {\tiny{1}};
    \node at (5.5*45:\radiux) {\tiny{1}};
    \node at (7.5*45:\radiux) {\tiny{1}};
    \node [blk] at (P0) {};
    \node [wht] at (P1) {};
    \node [blk] at (P2) {};
    \node [wht] at (P3) {};
    \node [blk] at (P4) {};
    \node [wht] at (P5) {};
    \node [blk] at (P6) {};
    \node [wht] at (P7) {};
  \end{tikzpicture}
  \quad \longmapsto \quad M' \,=\!\!
  \begin{tikzpicture}[semithick, scale=0.65, bend angle=45, rotate=0,  baseline=-0.6ex] 
    \style;
    \path
    coordinate (A)  at (0*45:2*\radius) {}
    coordinate (B)  at (0*45:2.5*\radius) {}
    coordinate (Q)  at (1.7*45:\radius) {}
    coordinate (P0) at (0*45:\radius) {}
    coordinate (P1) at (1*45:\radius) {}
    coordinate (P2) at (2*45:\radius) {}
    coordinate (P3) at (3*45:\radius) {}
    coordinate (P4) at (4*45:\radius) {}
    coordinate (P5) at (5*45:\radius) {}
    coordinate (P6) at (6*45:\radius) {}
    coordinate (P7) at (7*45:\radius) {};
    \filldraw[fill=blue!6,draw=black!50] (0,0) circle (\radius);
    \node[font=\small] at (0,0) {$M$};
    \path (P6) edge[->, bend right, shorten >=2pt] (B);
    \path (P7) edge[bend right] (A);
    \path (B) edge[bend angle=90, bend right, looseness=1] (P2);
    \path (A) edge[bend angle=90, bend right, looseness=1] (Q);
    \path (A) edge (B);
    \path (P1) edge[bend left, looseness=1.5] (P2);
    \path (P3) edge[bend left, looseness=1.5] (P4);
    \path (P5) edge[bend left, looseness=1.5] (P6);
    \path (P7) edge[bend left, looseness=1.5] (P0);
    \draw[very thick, white] (Q) arc (1.7*45:2*45:\radius);
    \draw[very thick, white] (P6) arc (6*45:7*45:\radius);
    \node at (4.5*45:1.9) {\tiny{3}};
    \node at (0.5*45:1.9) {\tiny{3}};
    \node at (7.0*44.3:1.9) {\tiny{2}};
    \node at (0.5*45:3.17) {\tiny{1}};
    \node at (3.5*45:\radiux) {\tiny{1}};
    \node at (5.5*45:\radiux) {\tiny{1}};
    \node at (7.5*45:\radiux) {\tiny{1}};
    \node [blk] at (A)  {};
    \node [wht] at (B)  {};
    \node [blk] at (P0) {};
    \node [wht] at (P1) {};
    \node [blk] at (P2) {};
    \node [wht] at (P3) {};
    \node [blk] at (P4) {};
    \node [wht] at (P5) {};
    \node [blk] at (P6) {};
    \node [wht] at (P7) {};
    \node [circle,fill=white,draw=black, minimum size=2.5pt, inner sep=0pt] at (Q) {};
  \end{tikzpicture}
  $$
\end{itemize}
Any irreducible bicubic map on $n+2$ vertices can be constructed from some
bicubic map on $n$ vertices by applying operation (1) or (2) above.

We shall now describe how to create a reducible map based on
irreducible maps $M_1$, $M_2$, \ldots, $M_k$. An illustration for
$k=3$ can be found below. This corresponds to the $\oplus$-operation on \btrees.

\begin{itemize}
\item[(3)] We begin by lining up the maps $M_1$, $M_2$, \ldots,
  $M_k$. Next, in each map $M_i$, we remove the first edge (in {\em
    counter-clockwise direction}) from the root edge on the root
  face. Then we connect the maps as shown in the figure, and define
  the root edge of the obtained map to be the root edge of $M_k$.
\end{itemize}

\def\mapMi{\filldraw[fill=blue!6,draw=black!50] (0,0) circle (1);}
$$
\begin{tikzpicture}[semithick, scale=0.61, bend angle=45, rotate=0,  baseline=-0.6ex] %
  \style;
  \mapMi
  \draw [->] (180:1) arc (180:231:1);
  \draw      (240:1) arc (240:300:1);
  \node [blk] (r1) at (180:1) {};
  \node [wht] (s1) at (240:1) {};
  \node [blk] (t1) at (300:1) {};
  \path (s1) edge[bend left, looseness=1.5] (t1);
  \node [yshift=1.3] at ($ (s1) !.5! (t1) $) {\tiny{1}};
  \node[font=\small] at (0,0) {$M_1$};
  \node at (-1.2,0.7) {\tiny{3}};
  \begin{scope}[xshift=80pt]
    \mapMi
    \draw [->] (180:1) arc (180:231:1);
    \draw      (240:1) arc (240:300:1);
    \node [blk] (r2) at (180:1) {};
    \node [wht] (s2) at (240:1) {};
    \node [blk] (t2) at (300:1) {};
    \path (s2) edge[bend left, looseness=1.5] (t2);
    \node [yshift=1.3] at ($ (s2) !.5! (t2) $) {\tiny{1}};
    \node[font=\small] at (0,0) {$M_2$};
    \node at (-1.2,0.7) {\tiny{3}};
  \end{scope}
  \begin{scope}[xshift=160pt]
    \mapMi
    \draw [->] (180:1) arc (180:231:1);
    \draw      (240:1) arc (240:300:1);
    \node [blk] (r3) at (180:1) {};
    \node [wht] (s3) at (240:1) {};
    \node [blk] (t3) at (300:1) {};
    \path (s3) edge[bend left, looseness=1.5] (t3);
    \node [yshift=1.3] at ($ (s3) !.5! (t3) $) {\tiny{1}};
    \node[font=\small] at (0,0) {$M_3$};
    \node at (-1.2,0.7) {\tiny{3}};
  \end{scope}
\end{tikzpicture}
\quad \longmapsto \quad M' \,=\;
\begin{tikzpicture}[semithick, scale=0.61, bend angle=45, rotate=0,  baseline=-0.6ex] %
  \style;
  \mapMi
  \draw[->]           (180:1) arc (180:231:1);
  \draw[thick, white] (240:1) arc (240:300:1);
  \node [blk] (r1) at (180:1) {};
  \node [wht] (s1) at (240:1) {};
  \node [blk] (t1) at (300:1) {};
  \path (s1) edge[bend left, looseness=1.5] (t1);
  \node[font=\small] at (0,0) {$M_3$};
  \begin{scope}[xshift=80pt]
    \mapMi
    \draw               (180:1) arc (180:231:1);
    \draw[thick, white] (240:1) arc (240:300:1);
    \node [blk] (r2) at (180:1) {};
    \node [wht] (s2) at (240:1) {};
    \node [blk] (t2) at (300:1) {};
    \path (s2) edge[bend left, looseness=1.5] (t2);
    \node[font=\small] at (0,0) {$M_2$};
    \node at (0,-1.4) {\tiny{1}};
  \end{scope}
  \begin{scope}[xshift=160pt]
    \mapMi
    \draw               (180:1) arc (180:231:1);
    \draw[thick, white] (240:1) arc (240:300:1);
    \node [blk] (r3) at (180:1) {};
    \node [wht] (s3) at (240:1) {};
    \node [blk] (t3) at (300:1) {};
    \path (s3) edge[bend left, looseness=1.5] (t3);
    \node[font=\small] at (0,0) {$M_1$};
  \end{scope}
  \path (t1) edge[bend right=30] (s2)
        (t2) edge[bend right=30] (s3)
        (s1) edge[bend right=50, looseness=0.6] (t3);
  \node at (1.35,0.8) {\tiny{3}};
\end{tikzpicture}
$$
Any reducible bicubic map on $n$ vertices can be constructed by applying
the above operation (3) to some ordered list of irreducible bicubic maps whose
total number of vertices is $n$.

By defining operations on bicubic maps corresponding to the operations
$\lambda_i$ and $\oplus$ we have now completed the definition of the
bijection $\psi$ between bicubic maps and \btrees. See \cite{CKM} for examples of non-trivial applications of the bijection.

\section{Gray codes}

In this section, after introducing some notations, we will define a Gray code for 
\btrees\ with the same underlying tree and extend it to a cyclic Gray code for 
\btrees\ with arbitrary underlying trees. This will induce a cyclic Gray code for bicubic maps. Then we will see
that our construction can be easily extended to a wider class of $\beta$-description trees
inducing cyclic Gray codes for the corresponding to them planar maps.		

\medskip

A list $\mathcal L$ for a set of length $n$ tuples is a {\em Gray code} if $\mathcal L$ lists, with no repetitions nor omissions, the tuples in the set so that the Hamming distance between two successive tuples in $\mathcal L$
(i.e., the number of positions in which they differ) is bounded by a constant, independent of $n$.
And when we want to explicitely refer to this constant, say $k$, we call such a list a {\em $k$-Gray code}.
In addition, if the last and first tuple in $\mathcal L$ differ in the same way, then the Gray code is {\em cyclic}. 

\medskip


If $\mathcal L$  is a list, then $\overline {\mathcal L}$ is the list obtained by reversing $\mathcal L$, and
if ${\mathcal M}$ is another list, then ${\mathcal L}\circ {\mathcal M}$ is the
concatenation of the two lists.
If $\alpha$ is a tuple, then $\alpha\cdot{\mathcal L}$ 
(resp., ${\mathcal L}\cdot\alpha$) is the list obtained by appending (resp., postpending)
$\alpha$ to each tuple of $\mathcal L$, and $\alpha^k$ is the tuple obtained by concatening $k$
copies of $\alpha$.
Often we refer to a list by enumerating its elements, e.g. 
$\mathcal L=\langle \alpha_1,\alpha_2,\alpha_3,\ldots\rangle$.

Given a family $\{\mathcal L_1,\mathcal L_2,\ldots ,\mathcal L_m\}$
of $m$ lists, each $\mathcal L_i$, $1\leq i\leq m$, being a $k_i$-Gray code for a set $L_i$ of same length tuples, we define another family of lists $\{\mathcal N_1,\mathcal N_2,\ldots ,\mathcal N_m\}$ as follows:
$\mathcal N_1$ is simply the list $\mathcal L_m$,
and for $2\leq i\leq m$,

$$\mathcal N_i=e_1\cdot {\mathcal N_{i-1}}\circ e_2\cdot \overline{\mathcal N_{i-1}}\circ 
e_3\cdot\mathcal N_{i-1}\ldots,$$
where $\langle e_1,e_2,e_3,\dots ,e_j\rangle$ is the list $\mathcal L_{m-i+1}$,
and the last term of the concatenation defining $\mathcal N_i$ is either
$e_j\cdot {\mathcal N_{i-1}}$ or $e_j\cdot \overline{\mathcal N_{i-1}}$,
depending on $j$ being odd or even.

It is routine to prove the following proposition that we will use later.

\begin{proposition}
With the notations above, $\mathcal N_m$ is a $k$-Gray code for the product set 
$L_1\times L_2 \times \dots\times L_m$, where $k=\max\{k_1,k_2,\dots ,k_m\}$.
\label{product_set}
\end{proposition}

\subsection{Gray coding \btrees\ with the same underlying tree}\label{btrees-same-underlying}

Recall that by definitions, the label of the root of a \btree\ is uniquely determined by the labels of its children.
In what follows, for convenience, we assume that the root of any \btree\ is labeled by $*$. That is, we amend (2) in Definition~\ref{beta01-tree}, obtaining a class of trees, still called \btrees\ by us, which are in one-to-one correspondence with \btrees\ defined either in Definition~\ref{betaAB-tree} or Definition~\ref{beta01-tree}. The rationale for (again!) updating slightly our previous definitions is in allowing \btrees\ to be on distance 1 from each other in the sense specified in Definition~\ref{distance-same-underlying} below. 

We are also interested in the following class of labeled trees, which are essentially \btrees, but where the root is treated as any other internal vertex.

\begin{definition}\label{beta01'-tree} A $\beta'(0,1)$-tree is a rooted plane tree whose vertices are labeled with nonnegative integers such that
\begin{enumerate}
\item leaves have label $0$;
\item the label of any other vertex exceeds the sum of its children's
  labels by at most 1.
\end{enumerate}
\end{definition}

Note that any \btree\ discussed in Definition~\ref{beta01-tree} is a $\beta'(0,1)$-tree.

\begin{definition}\label{def-underl} Let $T$ be a \btree\ or a $\beta'(0,1)$-tree. Then  $u(T)$ denotes the underlying rooted tree, that is, the tree obtained by removing all labels in $T$. In other words, $u(T)$ gives the shape of $T$. \end{definition}

For example, if $T$ is the rightmost \btree\ in the top row in Figure~\ref{beta01}, then $u(T)$ is given by
$$\treb{ }{ }{ }{ }$$

\begin{definition} Let $T$ be a \btree\ or a $\beta'(0,1)$-tree with $n$ vertices. 
We let $\ell(T)$ denote the $n$-tuple of $T$'s labels obtained by traversing $T$ by depth first algorithm using the leftmost option and reading each label exactly once. \end{definition}

For example, for the tree $T$ in Figure~\ref{fig:tree}, $\ell(T)=(4,0,0,1,0,2,1,3,2,1,0,0)$;
see also Table \ref{co_example} where the roots are labeled by *.
Thus, $\ell(T)$ is an encoding of a given $\beta(0,1)$-tree $T$ in the form of a tuple.
Note that this encoding disregards the shape of the tree.

\begin{definition}\label{distance-same-underlying} For \btrees\ (with the root labeled by *), or $\beta'(0,1)$-trees, $T_1$ and $T_2$ such that $u(T_1)=u(T_2)$, the {\em distance} $d(T_1,T_2)$ between the trees is the number of positions in which $\ell(T_1)$ and $\ell(T_2)$ differ
(or equivalently, the Hamming distance between  $\ell(T_1)$ and $\ell(T_2)$). \end{definition}

For example, keeping in mind that we label the root of a \btree\ by $*$, the distance between the first and the forth trees in the top row in Figure~\ref{beta01} is 2, while the distance between the first and the second trees in the bottom row in that figure is 1.

\begin{definition}\label{def-L-T} For a $\beta'(0,1)$-tree $T$, we let $L(T)$ denote the set 
of encodings of all $\beta'(0,1)$-trees obtained by labeling properly $u(T)$.
\end{definition}

\begin{lemma}\label{lem1} For any $\beta'(0,1)$-tree $T$ there is a $1$-Gray code for $L(T)$. \end{lemma}
Before giving a formal proof of Lemma~\ref{lem1}, we explain its general idea, which is presented graphically in Figure~\ref{two_figures}(a). 
Assuming the existence of a $1$-Gray code for smaller trees, we can extend such a code to trees obtained by adding the new root. More precisely, each vertex in Figure~\ref{two_figures}(a) corresponds to a $\beta'(0,1)$-tree having a fixed shape (that is, underlying tree).
Vertices on the same vertical line correspond to $\beta'(0,1)$-trees that differ only in the root label: the higher vertex is, the larger root label it corresponds to. Note that each vertical line must contain at least the vertices corresponding to root labels~$0$ and~$1$, but it may or may not contain other vertices. 

Further, $a_1$, $a_2$, etc. in this figure form a Hamiltonian path corresponding to the 1-Gray code for the trees with root label equal~$0$, whose existence we assumed. Also, $b_1$, $b_2$, etc. is such a path for the trees with root label equal~$1$. Thus, the trees corresponding to $a_i$ and $b_i$ differ only in the root label. The desired Hamiltonian path
through all the $\beta'(0,1)$-trees (corresponding to the 1-Gray code for $L(T)$)
presented schematically in Figure~\ref{two_figures}(a) begins at $a_1$ and goes in the direction of the arrow.   \\

\noindent
{\it Proof of Lemma \ref{lem1}.}
We proceed by induction on the number $v$ of vertices in $T$. The base cases, 
$v=1$ (the single vertex $\beta'(0,1)$-tree) and $v=2$ (two single edge $\beta'(0,1)$-trees, which are on distance 1 from each other) obviously hold.  

Suppose now that $v\geq 3$,  and the children of the root are the roots of subtrees $T_1,T_2,\ldots,T_k$ from left to right, where $k\geq 1$. Each $T_i$ is a $\beta'(0,1)$-tree and by induction hypothesis, $L(T_i)$ has a  $1$-Gray code. But then, by Proposition~\ref{product_set},
$$L(T_1)\times L(T_2)\times \cdots \times L(T_k)$$ 
also has a $1$-Gray code, and it can be extended to a $1$-Gray code of $L(T)$  obtained by adding the new leftmost coordinate corresponding to $T$'s root to each entry of $L(T_1)\times L(T_2)\times \cdots \times L(T_k)$ as explained below. 

For an integer $u\geq 0$ we define two lists of $1$-tuples:
\begin{itemize}
\item $\gamma(u)=\langle (0),(u),(u-1),\dots, (2),(1)\rangle$, and
\item $\delta(u)=\langle (1),(u),(u-1),\dots, (2),(0)\rangle$.
\end{itemize}
In particular,  $\gamma(1)$
and $\delta(1)$ are the lists $\langle (0), (1)\rangle$ and 
$\langle (1), (0)\rangle$, respectively.

Let $\langle \alpha_1,\alpha_2,\ldots  \rangle$ be the $1$-Gray code list 
for $L(T_1)\times L(T_2)\times \cdots \times L(T_k)$, so that 
each $\alpha_j$ is the concatenation of $k$ tuples corresponding to the labels of 
the vertices of the trees $T_1,T_2,\ldots,T_k$, and let $m(\alpha_j)$ be the sum 
of the labels of the roots of these trees plus one.
In other words, $m(\alpha_j)$ is the maximal value of $x$, such that 
$(x)\cdot \alpha_j$ is a proper labeling of $T$. Thus $m(\alpha_j)\geq 1$
and  $m(\alpha_j)=1$ if and only if the root of each $T_i$ is labeled by $0$.
Finally, let $\mathcal M$ be the list defined as 

$$
\mathcal M  =  \mathcal M_1\circ  \mathcal M_2\circ \mathcal M_3\circ\cdots
$$
\noindent
with 
$$
\mathcal M_j=\left\{ \begin {array}{cccc}
\gamma(m(\alpha_j))\cdot\alpha_j	& {\rm if\ } j {\rm\ is\ odd},\\
\delta(m(\alpha_j))\cdot\alpha_j	& {\rm if\ } j {\rm\ is\ even}.       
\end {array}
\right.
$$

Clearly, the underlying set of $\mathcal M$ is $L(T)$. In addition $\mathcal M$ is a $1$-Gray code:
throughout each list $\mathcal M_j$ successive tuples differ in the first position, and the last
element of $\mathcal M_j$ differ from the first element of $\mathcal M_{j+1}$ as 
$\alpha_j$ differ from $\alpha_{j+1}$, that is in a single position.

Thus, $L(T)$ has a $1$-Gray code and the statement is proved by induction. 
\hfill \ensuremath{\Box}

\begin{theorem}\label{Z(T)-Hamilt-path} There exists a $1$-Gray code for \btrees\ having the same underlying tree. \end{theorem}

\begin{proof} Suppose that the root of a \btree\ $T$, labeled by $*$, has subtrees $T_1,\ldots,T_k$, where $k\geq 1$. Each $T_i$ is a $\beta'(0,1)$-tree, and thus, by Lemma~\ref{lem1}, there is a $1$-Gray code for each $L(T_i)$. But then, by Proposition \ref{product_set},
there is also a $1$-Gray code for $L(T_1)\times L(T_2)\times \cdots \times L(T_k)$ leading to the fact that $L(T)=\{(*)\}\times L(T_1)\times L(T_2)\times \cdots \times L(T_k)$ has a $1$-Gray code, as desired.\end{proof}

\begin{figure}
\begin{tabular}{cc}
\psfig{figure=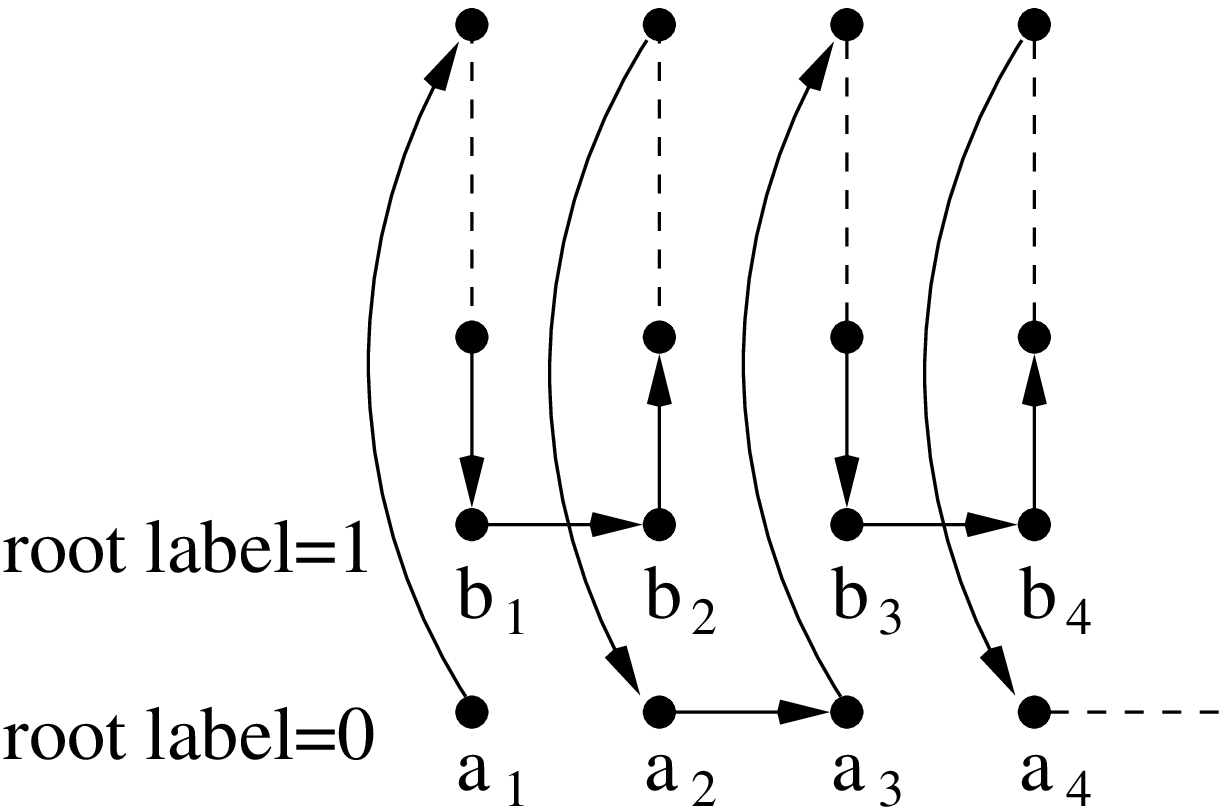,height=3cm}    &
\psfig{figure=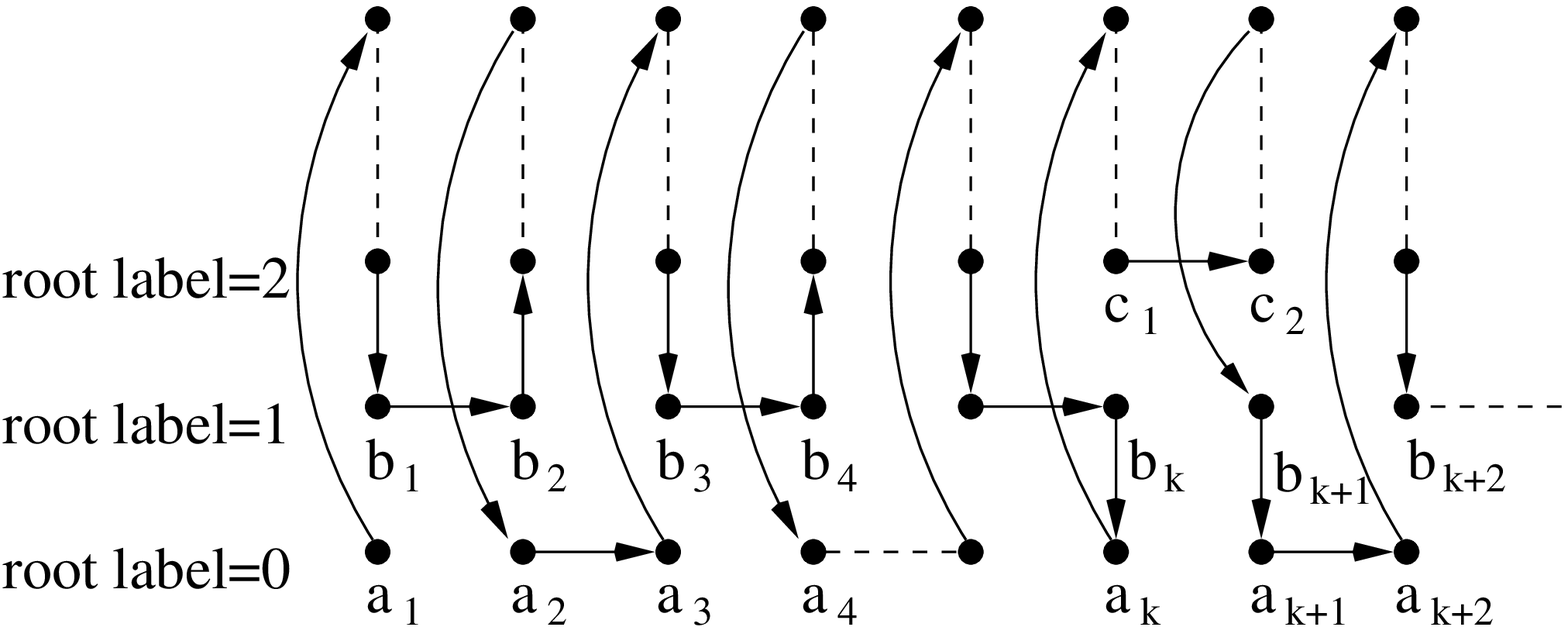,height=3cm}   \\
   (a) & (b)  \\
\end{tabular}
\caption{\label{two_figures}
Schematic approach in the proof of (a) Lemma \ref{lem1}, and (b) Lemma \ref{lem1-gen}.}
\end{figure}

Note that generally speaking the 1-Gray codes in Lemma \ref{lem1} and Theorem \ref{Z(T)-Hamilt-path} are not cyclic, and 
below we discuss how cyclic 1-Gray codes in this context can be obtained.

\subsection{Gray coding cyclicly \btrees\ with the same underlying tree}\label{btrees-same-underlying-cyclic}

The results presented in the previous subsection can be generalized to cyclic Gray codes in question. The goal of this subsection is to justify this, and in contrast with the previous subsection, 
here the proofs will be rather existential than constructive.
First note that Proposition~\ref{product_set} can be generalized to the following proposition that is easy to prove directly, but also it follows from more general results presented in~\cite{DPP}. 

\begin{proposition}
Suppose that $L_1,L_2,\ldots,L_m$ are sets, and 
each of them is a set of same length tuples which is 
either a singleton, or there is a cyclic $1$-Gray code for it. Then there is a cyclic $1$-Gray code for the product set 
$L_1\times L_2 \times \dots\times L_m$.
\label{product_set_cyclic}
\end{proposition}

In what follows we will need the following easy to understand facts. 

\begin{Fact} 
\label{fact}
Let $\mathcal G$ be a $1$-Gray code for a set of length $k$ tuples.
\begin{itemize}
\item If $r=(r_1,\ldots ,r_k)$ and $s=(s_1,\ldots ,s_k)$ are consecutive tuples in $\mathcal G$,
then there are no $i$ and $j$,  $1\leq i\neq j\leq k$, such that $r_i\neq s_i$ and $r_j\neq s_j$.
\item If $t=(t_1,\ldots ,t_k)$ is a third tuple, and $r$, $s$ and $t$ are consecutive 
(in this order) in $\mathcal G$ and there is an $i$ with $s_i\notin \{r_i,t_i\}$, then
$r_i\neq t_i$.
\end{itemize}
\end{Fact}

As mentioned before, the Gray code defined in Lemma \ref{lem1} is not necessarily cyclic.
However, if the number of $a_i$s in the construction of the Hamiltonian path $P$ in Figure \ref{two_figures}(a) is even, then the assumption of existence of a Hamiltonian cycle (or equivalently, of a cyclic $1$-Gray code) for smaller $\beta'(0,1)$-trees would give us the cyclicity of the path $P$, and so, the cyclicity of the Gray code defined in Lemma \ref{lem1}.
Indeed, in this case the first and the last vertices in $P$ are $a_1$ and $a_m$, where $m$ is 
the maximum index, and $a_1$ and $a_m$ are connected by an edge.  
 
Thus, the difficult case is when the number of $a_i$s is odd, which is possible (e.g. there are five $\beta'(0,1)$-trees having the shape of a path on three vertices).
To handle this situation, we use the idea presented in Figure~\ref{two_figures}(b). Namely, we will prove that essentially in all cases, there is an edge $c_1c_2$ between two $\beta'(0,1)$-trees having root label 2 such that the edges corresponding to it on levels ``root label = 0''  and ``root label = 1'' are involved in the respective Hamiltonian cycles assumed by the induction hypothesis (these edges are $a_ka_{k+1}$ and $b_kb_{k+1}$ in Figure~\ref{two_figures}(b)). The existence of $c_1c_2$ allows us to change the turns taken from the two Hamiltonian cycles making sure that the first and the last vertices in $P$ will be $a_1$ and $a_m$. 
In what follows, for a given tree we will refer to a vertex different from the root or a leaf as {\it internal vertex}.
 The two situations when  $c_1c_2$ does not exist are easy to handle. This happens when 

\begin{itemize} 
\item there are no internal vertices in a tree. For a given number of vertices, there are only two such $\beta'(0,1)$-trees, any listing of which gives a cyclic $1$-Gray code;
\item there is exactly one internal vertex in a tree. One can easily check that, for any tree $T$ with exactly one internal vertex, there are five $\beta'(0,1)$-trees with the shape of $T$,
and the number of such trees with root label 0 (equivalently, 1), which is 
the number of $a_i$'s (equivalently, $b_i$'s)
is two, which is even, so that the existence of a Hamiltonian cycle in this case is easy to establish using the approach in Figure~\ref{two_figures}(a).  For example, the (ordered) list below is a cyclic $1$-Gray code for the $\beta'(0,1)$-trees of same shape on four vertices:
$$  
\treb{0}{0}{0}{0}\qquad
\treb{0}{0}{0}{1}\qquad
\treb{0}{0}{1}{1}\qquad
\treb{0}{0}{1}{2}\qquad
\treb{0}{0}{1}{0}\qquad
$$
\end{itemize}

Recall from Definition \ref{def-L-T} that for a $\beta'(0,1)$-tree $T$, $L(T)$ denotes the set 
of encodings of all $\beta'(0,1)$-trees obtained by labeling properly $u(T)$. The following lemma generalizes Lemma~\ref{lem1}.

\begin{lemma}\label{lem1-gen} For any $\beta'(0,1)$-tree $T$ there is a cyclic $1$-Gray code for $L(T)$. \end{lemma}

\begin{proof} Based on the discussion preceding the statement of this lemma, the proof of Lemma~\ref{lem1} can be used if we will prove the existence of an edge $c_1c_2$ in the situation when $T$ has at least two internal vertices. 
If the subtrees of $T$ are $T_1,T_2\ldots,T_k$, the existence of $c_1c_2$ is equivalent with the existence of two successive tuples in the Gray code for the product set $L(T_1)\times L(T_2)\times \cdots \times L(T_k)$
which both can be extended to proper labelings of $T$ by letting the root of $T$ be 2.
To this end, we consider two subcases.

\begin{itemize}
\item The root of $T$ has exactly one internal vertex $v$ among its children. It follows that, 
$T_j$, the subtree rooted in $v$ in turn has at least one internal vertex among its children.
It is easily seen that the product set $L(T_1)\times L(T_2)\times \cdots \times L(T_k)$
collapses to $(0,\ldots,0) \cdot L(T_j)\cdot (0,\ldots,0)$, for appropriate length 
$0$s prefix and suffix, and a Gray code for this product set is essentially a Gray code for  $L(T_j)$.
Let $u$ be a tuple in $L(T_j)$ representing a proper labeling of $T_j$ 
where the label of $v$ (that is, the first entry of $u$) is a non-zero value. Since $T_j$ has at least one internal vertex it follows that there are at least three such $u$, and we choose one which is not the last nor the first tuple in the Gray code for $L(T_j)$. 
The existence of this Gray code follows by inductive hypothesis.
By the second point of Fact \ref{fact}, at least one among the successor and the predecessor of $u$, in the Gray code for $L(T_j)$, has a non-zero value in the position corresponding to $v$, the root of $T_j$.
It follows that there are two successive tuples in the Gray code for $L(T_j)$, and so for the 
above product set, which can be extended to a proper labelings of $T$ by letting the root of $T$ be 2, which is allowed by the rules of $\beta'(0,1)$-trees; and this extension will give us the desired edge $c_1c_2$.

\item The root of $T$ has at least two internal vertices among its children, and let 
$i$ and $j$ be the positions corresponding to the labels of two of these vertices in the tuples in
$L(T_1)\times L(T_2)\times \cdots \times L(T_k)$.
Clearly, $1$ is an admissible value for the entries in positions $i$ and $j$, and let $u$ be
a tuple in $L(T_1)\times L(T_2)\times \cdots \times L(T_k)$ having $1$ in both positions $i$ and 
$j$.
By the first point of Fact~\ref{fact} it follows that if $u$ is not the last tuple
in the  Gray code for the product set $L(T_1)\times L(T_2)\times \cdots \times L(T_k)$ 
(whose existence we assume by inductive hypothesis and Proposition \ref{product_set_cyclic}), then the successor of  $u$ has $1$ in at least one of these two positions.
The reasoning is similar when $u$ is the last tuple by replacing ``successor'' by ``predecessor''.
And again, it follows that there are two successive tuples in the Gray code for $L(T_1)\times L(T_2)\times \cdots \times L(T_k)$ which can be extended to proper labelings of $T$ by letting the label of the root of $T$ be 2.    
\end{itemize}
Thus, in both cases there is a $1$-Gray code for $L(T)$ with the first and last tuple of the form $(0,\alpha_1,\alpha_2,\dots)$ 
and $(0,\omega_1,\omega_2,\dots)$, respectively, where$(\alpha_1,\alpha_2,\dots)$ and $(\omega_1,\omega_2,\dots)$ are the first and last tuples in the cyclic Gray code for the product set $L(T_1)\times L(T_2)\times \cdots \times L(T_k)$, and the statement follows.
\end{proof}

Using Lemma~\ref{lem1-gen}, we can now generalize Theorem~\ref{Z(T)-Hamilt-path}, which is the main result in this subsection.

\begin{theorem}\label{Z(T)-Hamilt-path-cycle} There exists a cyclic $1$-Gray code for \btrees\ having the same underlying tree. \end{theorem}

\begin{proof} Suppose that the root of a \btree\ $T$, labeled by $*$, has subtrees $T_1,\ldots,T_k$, where $k\geq 1$. Each $T_i$ is a $\beta'(0,1)$-tree, and thus, by Lemma~\ref{lem1-gen}, there is a cyclic $1$-Gray code for each $L(T_i)$. But then, by Proposition~\ref{product_set_cyclic},
there is also a cyclic $1$-Gray code for $L(T_1)\times L(T_2)\times \cdots \times L(T_k)$ leading to the fact that $L(T)=\{(*)\}\times L(T_1)\times L(T_2)\times \cdots \times L(T_k)$ has a cyclic $1$-Gray code, as desired.\end{proof}

Since the tuple $(*,0,0,\ldots ,0)$ of appropriate length is always an admissible encoding of a $\beta(0,1)$-trees, 
we have:

\begin{corollary}
\label{one_corr}
There exists a $1$-Gray code list for \btrees\ having the same underlying tree which begins by 
$(*,0,0,\ldots ,0)$ and ends by a tuple differing from $(*,0,0,\ldots ,0)$ in exactly one position.
\end{corollary}

\subsection{Dyck words and Gray coding bicubic maps}\label{main}

For two integers $k$ and $m$, $0\leq k\leq  m$, we denote by $D_{m,k}$ the set of binary tuples 
with $m$ occurrences of $1$ and  $k$ occurrences of $0$,
satisfying the {\it prefix property}: no prefix contains more $0$s than $1$s. For example, $(1,0,1,1,0,1,1,0,0)\in D_{5,4}$. The set $D_{m,m}$ is known as the set of {\em Dyck words} of length $2m$. The number of elements in $D_{m,m}$ is the well-known $m$th {\em Catalan number} $C_{m}=\frac{1}{m+1}{2m \choose m}$.
See Table \ref{list_pref} for the 14 length 8 Dyck words.

Traversing a plane tree by the depth first algorithm using the leftmost option and letting 1 represent a forward step, while 0 represent a backward step, we obtain a one-to-one correspondence between plane trees on $n+1$ vertices (and $n$ edges) and Dyck words of length $2n$. We let $w(T)$ denote the Dyck word corresponding to a plane tree $T$ under this bijection. For example, for the tree $T$ on 4 vertices illustrating Definition~\ref{def-underl}, $w(T)=(1,1,0,1,0,0)$;
see also Figure \ref{co_example}.

Note that the Hamming distance between two same length Dyck words is always even, and thus the minimum distance between two Dyck words is 2. 

The following recursive description obtained in \cite{Vaj_99} gives 
a Gray code for $D_{m,k}$, and, in particular, for the set of length $2m$
Dyck words; this description is a slight variation of the code defined in~\cite{Rus_90}:

\begin{equation}
\mathcal{D}_{m,k}=\left\{ \begin {array}{ccc} 
     (1)^m         & \text{if} & k=0,  \\   
    \mathcal{D}_{m,k-1}\cdot (0)  & \text{if} & m=k>0,\\
    \mathcal{D}_{m-1,k}\cdot (1) \circ  \overline {\mathcal{D}}_{m,k-1}\cdot (0) & \text{if} & m>k>0.\\
\end {array}
\right.
\label{G_Dy}
\end{equation}
See Table \ref{list_pref}, showing $\mathcal{D}_{4,4}$, for an example. The following lemma was proved in~\cite{Vaj_99}.  

\begin{lemma}\cite{Vaj_99} The list $\mathcal{D}_{m,k}$ satisfies the following properties:
\begin{itemize}
\item The first tuple in $\mathcal{D}_{m,k}$ is $(1,0)^k\cdot (1)^{m-k}$; 
\item The last tuple in $\mathcal{D}_{m,k}$ is
\begin{itemize}
\item $(1,0)^{m-2}\cdot(1,1,0,0)$, if $k=m>1$,
\item $(1,0)^{k-1}\cdot(1)^{m-k+1}\cdot(0)$, if $m>k\geq1$ or $m=k=1$,
\item $(1)^m$, if $m>k=0$;
\end{itemize}
\item Two successive tuples in $\mathcal{D}_{m,k}$, including the last and the first one, differ in exactly two positions, and thus $\mathcal{D}_{m,k}$ is a cyclic $2$-Gray code.
\end{itemize}
\end{lemma}

\begin{table}
\begin{center}
\begin{tabular}{|c|}
\hline
(1, 0, 1, 0, 1, 0, 1, 0) \\
(1, 1, 0, 0, 1, 0, 1, 0) \\
(1, 1, 1, 0, 0, 0, 1, 0) \\
(1, 1, 0, 1, 0, 0, 1, 0) \\
(1, 0, 1, 1, 0, 0, 1, 0) \\
(1, 0, 1, 1, 1, 0, 0, 0) \\
(1, 1, 0, 1, 1, 0, 0, 0) \\
(1, 1, 1, 0, 1, 0, 0, 0) \\
(1, 1, 1, 1, 0, 0, 0, 0) \\
(1, 0, 1, 1, 0, 1, 0, 0) \\
(1, 1, 0, 1, 0, 1, 0, 0) \\
(1, 1, 1, 0, 0, 1, 0, 0) \\
(1, 1, 0, 0, 1, 1, 0, 0) \\
(1, 0, 1, 0, 1, 1, 0, 0) \\
\hline
\end{tabular}
\end{center}
\caption{
\label{list_pref}The Gray code list $\mathcal{D}_{4,4}$ for the set of length $8$ Dyck words.
}
\end{table}

In what follows, the Hamming distance between tuples is denoted by $d$, and the next definition 
extends it to trees.

\begin{definition}\label{distance-beta01-trees1} For \btrees\ $T_1$ and $T_2$ on the same number of vertices, the {\em distance} $d(T_1,T_2)$ between the trees is defined as
$$d(T_1,T_2)=d(\ell(T_1),\ell(T_2)) + d(w(u(T_1)), w(u(T_2))).$$
\end{definition}

\begin{theorem}\label{thm-main} There exists a cyclic $3$-Gray code for \btrees\ {\rm(}with the root labeled by *{\rm)} on $n$ vertices, $n\geq 1$,
with respect to the distance given in Definition~\ref{distance-beta01-trees1}.\end{theorem}

\begin{proof} Any \btree\ $T$ on $n$ vertices can be encoded by a $(3n-2)$-tuple, which is obtained by merging the $(2n-2)$-tuples $w(u(T))$ and the $n$-tuples $\ell(T)$. 
Let 
\begin{equation}
\label{list_gen}
d_1\cdot {\mathcal L}(T(d_1))\circ 
d_2\cdot {\mathcal L}(T(d_2))\circ
d_3\cdot {\mathcal L}(T(d_3))\circ
\cdots
\end{equation}
be the list where $d_i$ is the $i$th tuple in the list  $\mathcal{D}_{n-1,n-1}$ defined in relation (\ref{G_Dy}),
$T(d_i)$ is the tree encoded by $d_i$, and 
$\mathcal{L}(T(d_i))$ the list assumed by Corollary \ref{one_corr} for the  \btrees\ having the shape $T(d_i)$.

It is easy to see that in the list defined in relation (\ref{list_gen}) two successive tuples 
are at distance at most 3. Indeed,
\begin{itemize}
\item for a fixed $d_i$, successive tuples in $d_i\cdot {\mathcal L}(T(d_i))$ differ in one position, and 
\item for two successive tuples $d_i$ and $d_{i+1}$ in $\mathcal{D}_{n-1,n-1}$
(including the last and the first ones), the last tuple in $d_i\cdot {\mathcal L}(T(d_i))$ and the first one in $d_{i+1}\cdot {\mathcal L}(T(d_{i+1}))$ differ in three positions.
\end{itemize}
\end{proof}

Note that the Gray code stated in Theorem \ref{thm-main} for
$\beta(0,1)$-trees is minimal, in the sense that, in general there are not cyclic $2$-Gray codes for 
$\beta(0,1)$-trees. 
See Table~\ref{co_example} for an example when $\beta(0,1)$-trees, encoded by 
(1,1,0,0,*,0,0), 
(1,1,0,0,*,1,0)  and 
(1,0,1,0,*,0,0), cannot be listed cyclically 
so that the distance between successive trees is at most 2.
Also, the Gray code defined in (\ref{list_gen}) is ``shape partitioned'', that is, same shape $\beta(0,1)$-trees are successive in it.

\begin{table}[h]
\begin{tabular}{|c||M{1.1cm}|M{1.1cm}|M{1.1cm}|}
\hline
   $T$ & 
  \treav{0}{0}{*}\qquad &
  \treav{0}{1}{*}\qquad &
  \trebv{0}{0}{*}\qquad \tabularnewline\hline
$\ell(T)$ & (*,0,0) & (*,1,0) & (*,0,0)\tabularnewline\hline
$w(u(T))$ & (1,1,0,0) & (1,1,0,0)&(1,0,1,0)\tabularnewline\hline
\end{tabular}

\caption{
\label{co_example}
The three $\beta(0,1)$-trees on three vertices with the root labeled by *, and the corresponding to them the depth first leftmost option reading of the labels, and the 
Dyck words coding their shape.}
\end{table}

Our way to Gray code bicubic maps can be applied to any class of planar maps that can be described in terms of $\beta(a,b)$-trees with $b\geq 1$. Namely, generalizing the notion of $\beta'(0,1)$-trees to that of $\beta'(a,b)$-trees (by removing the condition on the root in Definition~\ref{betaAB-tree}), we can essentially copy/paste all our arguments for  $\beta'(0,1)$-trees. Indeed, for such a $\beta'(a,b)$-tree $T$, the levels corresponding to the root's label $a$ and $a+1$ will be isomorphic, so that induction can be used in the way we used it for  $\beta'(0,1)$-trees. Thus, in particular, we can Gray code $3$-connected cubic planar maps and cubic non-separable planar maps corresponding to $\beta(1,1)$-trees and $\beta(2,2)$-trees, respectively \cite{CJS,CS1997}. 

Finally, note that having $b\geq 2$ would simplify some of our arguments. In particular, in this case there is no need to prove the existence of the edge $c_1c_2$ in Lemma~\ref{lem1-gen}, since  we will have at least three isomorphic levels of vertices corresponding to the root labels $a$, $a+1$ and $a+2$, so that existence of an edge with the right properties will be given to us automatically (in fact, each edge from the Hamiltonian path on level $a+2$ will have the right properties).

\section{Concluding remarks%
}\label{final-sec}

In this paper we have shown that classes of planar maps corresponding to  $\beta(a,b)$-trees with $b\geq 1$ have cyclic 3-Gray codes, and
these codes are minimal in the sense of Hamming distance. We leave it as an open problem to determine whether there exist (cyclic) $k$-Gray codes, for some $k\geq 1$, for $\beta(a,0)$-trees, where $a\geq 1$. In the case $a=1$ such a code would induce a Gray code on non-separable planar maps via the respective bijection \cite{CJS}. 

It would also be interesting to Gray code so-called $\alpha$-description trees (see  \cite{CJS} for the definition) that would induce Gray coding of {\em planar maps} and {\em Eulerian planar maps} \cite{CJS}.

\section*{Acknowledgments}

The first and the third authors were supported by Grant NSh-1939.2014.1 of President of
Russia for Leading Scientific Schools. The second author is grateful to London Mathematical Society and to the University of Bourgogne for supporting his work on this paper.

\end{document}